\documentclass[a4paper,12pt]{article}
\usepackage{graphicx}
\usepackage{amsmath}
\usepackage{amssymb}
\usepackage{amsthm}
\theoremstyle{plain}
\newtheorem*{theorem}{Theorem}

\numberwithin{equation}{section}

\title{Simple Equations methodology (SEsM)  for searching of multisolitons and other exact solutions of nonlinear partial differential equations}
\author{Nikolay K. Vitanov$^{1,2}$, Zlatinka I. Dimitrova$^{3}$}
\date{$^1$ Institute of Mechanics, Bulgarian Academy of Sciences, Acad. G. Bonchev Str., Bl. 4, 1113, Sofia, Bulgaria\\
$^2$ Max-Planck Institute for the Physics of Complex Systems, Noethnitzer Str. 38, 01187 Dresden, Germany \\
$^3$ "Georgi Nadjakov" Institute of Solid State Physics, Bulgarian Academy of Sciences, Tsarigrasko Chaussee 72 Blvd., 1784 Sofia, Bulgaria}

\begin{document}

\maketitle

\begin{abstract}
We discuss a version the methodology for obtaining exact solutions of nonlinear partial differential equations based on  the  possibility for use of: (i) more than one simplest equation;
(ii) relationship that contains as particular 
cases the relationship used by Hirota \cite{hirota} and the relationship used in the
previous version of the methodology; (iii) transformation of the solution that contains as
particular case the possibility of use of the Painleve expansion; (iv) more than one balance equation.
The discussed version of the methodology allows: (i)
obtaining multi-soliton solutions of nonlinear partial differential
equations if such solutions do exist; (ii) obtaining particular solutions of
nonintegrable nonlinear partial differential equations. Several examples for the 
application of the methodology are discussed. Special attention is devoted to the 
use of the simplest equation $f_\xi =n[f^{(n-1)/n} - f^{(n+1)/n}]$ where $n$ is a
positive real number. This simplest equation allows us to obtain exact solutions
of nonlinear partial differential equations containing fractional powers.

\end{abstract}


\section{Introduction}
Differential equations relate quantities to their changes and such relationships are frequently 
encountered. Because of this differential equations occur in the mathematical analysis of many  
problems from natural and social sciences. Models based on nonlinear differential equations are 
used for description of processes in  fluid mechanics, solid state physics, plasma physics, atmospheric and ocean sciences, chemistry, materials science, mathematical biology, economics, social dynamics, etc. \cite{debn} - \cite{vsd2}. Numerical solution of the model nonlinear differential equations is always
an option but the search for exact analytical solutions of these equations is also important. 
Often the model equations are nonlinear partial differential equations (nonlinear PDEs) and by means of 
the exact solutions of these equations one can understand complex nonlinear phenomena such as existence 
and change of different regimes of functioning of complex systems, spatial localization, transfer 
processes, etc. In addition the exact solutions can be used to test computer programs for numerical 
simulations by comparison of the obtained numerical results
with the corresponding exact solutions. Because of all above the exact solutions of nonlinear partial differential equations are studied very intensively \cite{ablowitz1} - \cite{tabor}. In the course of
the time the methodology of obtaining exact solutions of nonlinear partial differential equations became  more and more complicated. In the yearly years of the research on the methodology one has searched
for transformations that can transform the solved nonlinear partial differential equation to a
linear differential equation. For an example the Hopf-Cole transformation \cite{hopf}, \cite{cole}  transforms the nonlinear Burgers equation to the linear heat equation. In the following years numerous 
attempts for obtaining such transformations have been made and in 1967 Gardner, Green, Kruskal and Miura 
\cite{gardner} managed to connect the Korteweg - de Vries equation to the inverse scattering problem for the linear Sch{\"o}dinger equation. This methodology was further developed and today it is a powerful methodology know as \emph{Method of Inverse Scattering Transform} \cite{ablowitz2}, \cite{ac}. Another line of research on
the exact solutions of nonlinear partial differential equations was followed by Hirota who developed a
direct method for obtaining of such exact solutions - \emph{Hirota method} \cite{hirota}, \cite{hirota1}.
Hirota method is based on bilinearization of the solved nonlinear partial differential equation by means of appropriate transformations. Truncated Painleve expansions may lead to many of these appropriate transformations \cite{tabor}, \cite{ct1} - \cite{wtk}.
\par 
An important development happened when Kudryashov \cite{k3} studied the possibility
for obtaining exact solutions of nonlinear partial differential equations on the basis of
a truncated Painleve expansion where the truncation happens after the "constant term" (i.e.,
the constant term is kept in the expansion). In the course of these studies Kudryashov mentioned that
a particular exact solution of certain nonlinear partial differential equations  can be
represented as power series of solutions of the Riccati equation or of equations for the 
elliptic functions of Jacobi and Weierstrass.  Further research of Kudryashov in this direction 
leaded to the formulation of the \emph{Method of Simplest Equation (MSE)} \cite{k05}. The method is based on determining of singularity order $n$ of the solved nonlinear partial differential equation.
Then a particular solution of the equation is searched as series containing powers of solutions
of a simpler equation called \emph{simplest equation}. In 2008 Kudryashov and Loguinova \cite{kl08}
extended the methodology to solutions of nonlinear ordinary differential equations constructed as 
power series  of solution of a simplest ordinary differential equation and applied this methodology
for obtaining traveling wave solutions of nonlinear partial differential equations. We refer to the 
articles of Kudryashov and co-authors for further results connected to \emph{MSE} \cite{k5a} - \cite{k15}.
\par 
Our contribution to the method of simplest equation began by the use of the ordinary differential 
equation of Bernoulli as simplest equation \cite{v10} and by application of the method to ecology
and population dynamics \cite{vd10} where we have used the concept of the balance equation.
Today the method of simplest equation has two versions:
\begin{description}
	\item[1.)]
	Version called \emph{Method of Simplest Equation - MSE} (the original version of the method by 
	Kudryashov) where the determination of the truncation  of the corresponding series of solutions of the 
	simplest equation is based on the first step in the algorithm for detection of the Painleve property.
	\item[2.)] 
	An equivalent version called \emph{Modified Method of Simplest Equation - MMSE} or \emph{Modified 
		Simple Equation Method - MSEM} \cite{kl08}, \cite{vdk}, \cite{v11} based on determination of the kind
	of the simplest equation and truncation of the series of solutions of the simplest equation by means 
	of application of a balance equation. Up to now our contributions to the methodology and its application
	are connected to this version of the method \cite{v11a} - \cite{vdv17}. We note especially the article \cite{vdv15} where we have extended the methodology of the \emph{MMSE} to simplest equations of the class
	\begin{equation}\label{sf}
	\left (\frac{d^k g}{d\xi^k} \right)^l = \sum \limits_{j=0}^{m} d_j g^j
	\end{equation}
	where $k=1,\dots$, $l =1,\dots$, and $m$ and $d_j$ are parameters. The solution of Eq.(\ref{sf}) defines
	a special function that contains as particular cases, e.g.,: (i) trigonometric functions; (ii) hyperbolic functions;
	(iii) elliptic functions of Jacobi; (iv) elliptic function of Weierstrass.
\end{description}
\par
We believe in the large potential of \emph{MSE} and \emph{MMSE} and our goal is to extend this 
methodology in order to make it applicable to larger classes of nonlinear partial differential equations. The text below is organized as follows. In Sect.2 we formulate a new version the methodology \cite{yy1} - \cite{yyx}. This version
makes the methodology  capable to obtain multi-soliton solutions of nonlinear partial differential equations. Sect. 3 is devoted to
a demonstration that the version of the method formulated in Sect. 2 really can lead to multi-soliton solutions and to solutions of nonitegrable partial differential equations on the basis of use of
truncated Painleve expansions. In Sect. 4 we demonstrate application of the method for the case without
use of truncated Painleve expansion. In Sect. 5 we show the application of the method for the case
of selected simplest equation containing non-integer powers of the unknown function. Several concluding
remarks are summarized in Sect. 6.
\section{Formulation of the Simple Equations Method (SEsM)}
In the previous version of the modified method of simplest equation one has used a representation of the searched solution of a nonlinear partial differential equation as power series of a solution of a 
simplest equation. This approach was proved to work if one search, e.g., 
for solitary wave solutions or
kink solutions of the solved nonlinear partial differential equation but not in the case of
search for bisoliton, trisoliton, and multisoliton solutions. The reason for this is that
the current version of the modified method of simplest equation  is connected to the use of a single
simplest equation. If we allow for use of more than one simplest equation then the modified method of simplest equation can be formulated in a way that makes obtaining of multisoliton solutions possible. Below we formulate such a version of the methodology. The schema of
the old and the new version of the methodology is shown in Fig. 1.
\begin{figure}[!htb]
	\centering
	\includegraphics[scale=0.5]{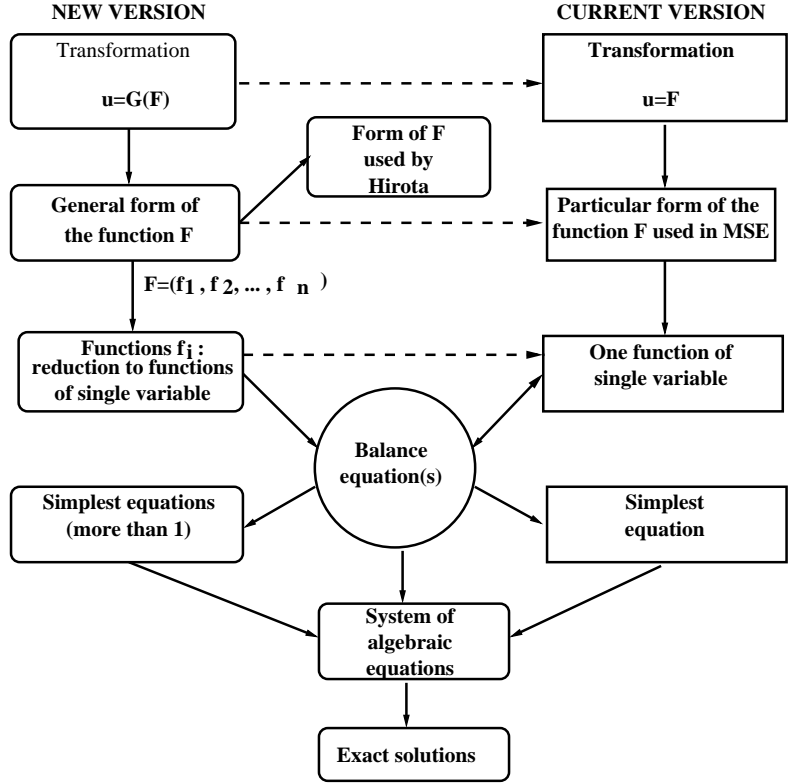}
	\caption{Schema of the current version (used up to now) and of the
		new version of the methodology. The transformation used in the new version of the method contains as particular case the transformation used in the current version of
		the method. The form of the function $F$ used in the new version
		of the method contains as particular case the form of the function
		used in the current version of the method. The new form of the function $F$ contains as particular case also the form of the function used by Hirota. In the new form of the method more than
		one simplest equation can be used and more than one balance equation can arise. In the current version of the method one uses one simplest equation and one balance equation.}
\end{figure}
\par
Let us consider a nonlinear partial differential equation 
\begin{equation}\label{eq}
{\cal{N}}(u,\dots)=0
\end{equation}
where ${\cal{N}}(u,\dots)$ depends on the function $u(x,...,t)$
and some of its derivatives participate in  ($u$ can be a function of more than 1 spatial coordinate).
The 7 steps of the methodology of the modified method of simplest equation are as follows.
\begin{description}
	\item[1.)]
	We perform a transformation
	\begin{equation}\label{m1}
	u(x,\dots,t)=G(F(x,\dots,t))
	\end{equation}
	where $G(F)$ is some function of another function  $F$. In general
	$F(x,\dots,t)$ is a function of the spatial variables as well as of the time. The transformation $G(F)$
	may be the Painleve expansion \cite{hirota}, \cite{kudr90}, \cite{k3}, \cite{w1} - \cite{k10} or another transformation, e.g., $u(x,t)=4 \tan^{-1}[F(x,t)]$ for the case of the 
	sine - Gordon equation or $u(x,t) = 4 \tanh^{-1}[F(x,t)]$ for the case of sh-Gordon (Poisson-Boltzmann 
	equation) (for applications of the last two transformations see, e.g. \cite{mv1} - \cite{mv5}).
	In many particular cases one may skip this step (then we have just $u(x,\dots,t)=F(x,\dots,t)$) 
	but in numerous cases the step is necessary
	for obtaining a solution of the studied nonlinear PDE. The application of Eq.(\ref{m1}) to 
	Eq.(\ref{eq}) leads to a nonlinear PDE for the function $F(x,\dots,t)$.
	\item[2.)]
	The function $F(x,\dots,t)$ is represented as a function of other functions $f_1,\dots,f_N$
	that are  connected to solutions of some differential equations (these equations can be partial 
	or ordinary differential equations) that are more simple than Eq.(\ref{eq}). We note that 
	the possible values of $N$ are $N=1,2,\dots$ (there may be infinite number of functions $f$ too).
	The forms of the function $F(f_1,\dots,f_N)$ can be different. 
	One example is
	\begin{eqnarray}\label{m2}
	F &=& \alpha + \sum \limits_{i_1=1}^N \beta_{i_1} f_{i_1} + \sum \limits_{i_1=1}^N  \sum \limits_{i_2=1}^N 
	\gamma_{i_1,i_2} f_{i_1} f_{i_2} + \dots + \nonumber \\
	&&\sum \limits_{i_1=1}^N \dots \sum \limits_{i_N=1}^N \sigma_{i_1,\dots,i_N} f_{i_1} \dots f_{i_N}
	\end{eqnarray}
	We shall use Eq.(\ref{m2}) below. Note that the relationship (\ref{m2}) contains as particular case the 
	relationship used by Hirota \cite{hirota}. The power series $\sum \limits_{i=0}^N \mu_n f^n$ (where
	$\mu$ is a parameter) used in the previous versions of the methodology of the modified method of simplest equation are a particular case of the relationship (\ref{m2}) too.
	\item[3.)] 
	In general the functions $f_1,\dots,f_N$ are solutions of partial differential equations.
	By means of appropriate ans{\"a}tze (e.g.,  traveling-wave ans{\"a}tze such as 
	$\xi = \hat{\alpha} x + \hat{\beta} t$; $\zeta =\hat{\gamma} x + \hat{\delta} t$, 
	$\eta = \hat{\mu} y + \hat{\nu}t \dots$) 
	the solved   differential equations for $f_1,\dots,f_N$ may be reduced to   differential equations 
	$E_l$, containing derivatives of one or several functions
	\begin{equation}\label{i1}
	E_l \left[ a(\xi), a_{\xi},a_{\xi \xi},\dots, b(\zeta), b_\zeta, b_{\zeta \zeta}, \dots \right] = 0; \ \
	l=1,\dots,N
	\end{equation}
	In many cases (e.g, if the equations for the functions $f_1,\dots$ are ordinary differential equations) one may skip this step 
	but the step may be necessary if the equations for $f_1,\dots$ are partial differential equations.
	\item[4.)]
	We assume that	the functions $a(\xi)$, $b(\zeta)$, etc.,  are  functions of 
	other functions, e.g., $v(\xi)$, $w(\zeta)$, etc., i.e.
	\begin{equation}\label{i1x}
	a(\xi) = A[v(\xi)]; \ \ b(\zeta) = B[w(\zeta)]; \dots
	\end{equation} 
	Note that the kinds of the functions $A$ , $B$, $\dots$ are not prescribed. 
	Often one uses a finite-series relationship, e.g., 
	\begin{equation}\label{i2}
	a(\xi) = \sum_{\mu_1=-\nu_1}^{\nu_2} q_{\mu_1} [v (\xi)]^{\mu_1}; \ \ \ 
	b(\zeta) = \sum_{\mu_2=-\nu_3}^{\nu_4} r_{\mu_2} [w (\zeta)]^{\mu_2}, \dots 
	\end{equation}
	where $q_{\mu_1}$, $r_{\mu_2}$, $\dots$ are coefficients.
	However other kinds of relationships may be used too. 
	\item[5.)]
	The functions  $v(\xi)$, $w(\zeta)$, $\dots$ 
	are solutions of simpler ordinary differential equations called \emph{simplest equations}. 
	For several years the methodology of the modified method of simplest equation was based 
	on use of one simplest equation. This version of the methodology allows for the use of more
	than one simplest equation. But these equation could be not the simplest possible ones. Because of this it is better to change the name of the methodology. As the equations are simple but not the simplest one we shall
	call the new version of the methodology \emph{Simple Equations Method (SEsM)}.
	\item[6.)]
	The application of the steps 1.) - 5.) to Eq.(\ref{eq}) transforms the left-hand side of 
	this equation. Let the result of this transformation  be a function that is a sum of terms where each 
	term contains some function multiplied by a coefficient. This coefficient contains some of the 
	parameters of the solved equation and some of the parameters of the solution. In the most cases
	a balance procedure must be applied in order to ensure that the above-mentioned relationships
	for the coefficients contain more than one term ( e.g., if the result of the transformation 
	is a polynomial then the balance procedure has to ensure that the coefficient of each 
	term of the polynomial is a relationship that contains at least two terms).
	This balance procedure may lead to one or more additional relationships among the parameters 
	of the solved equation and parameters of the solution. The last relationships are called 
	\emph{balance equations}. 
	\item[7.)]
	We may obtain a nontrivial solution of Eq. (\ref{eq})  if all coefficients mentioned in Step 6.) are
	set to $0$. This condition usually leads to a system of nonlinear algebraic equations for the 
	coefficients of the solved nonlinear PDE and for the coefficients of the solution. Any nontrivial 
	solution of this algebraic system leads to a solution the studied  nonlinear partial differential 
	equation. Usually the above system of algebraic equations contains many equations that have to 
	be solved with the help of   a computer algebra system. 
\end{description}
Below we shall apply the Simple Equations Methodology (SEsM)  for two different cases:
(i) in presence of transformation $u=G(F)$ - Sect. 3 or (ii) application of particular case of transformation $u=G(F)$, i.e., $u=F$ - Sect.4.
\section{Applications of SEsM for the case of presence of nontrivial transformation of kind (\ref{m1})}
\subsection{Tutorial example: The bisoliton solution of the Korteweg-de Vries equation}
This is probably the most simple possible example and its purpose is just to  show that the
new version of the methodology can be used to search for multisoliton solutions. 
\par
We consider a version of the Korteweg - de Vries equation
\begin{equation}\label{a1}
u_t + \sigma u u_x + u_{xxx}=0
\end{equation}
where $\sigma$ is a parameter. The 7 steps of the application of the version of the modified 
method of simplest equation from Sect. 2 are as follows
\begin{description}
	\item[1.)] \emph{The transformation}\\
	We set $u=p_x$ in Eq.(\ref{a1}). The result is integrated  and we apply the transformation
	$p=\frac{12}{\sigma} (\ln F)_x$. The result is
	\begin{equation}\label{a2}
	FF_{tx} + FF_{xxxx} - F_t F_x  + 3 F_{xx}^2 - 4 F_x F_{xxx} =0
	\end{equation}
	\item[2.)] \emph{Relationship among $F(x,t)$ and two functions $f_{1,2}$ that will be connected below to two simplest equations}\\ 
	We shall use two functions
	$f_1(x,t)$ and $f_2(x,t)$ and the relationship for $F$ is assumed to be a particular case of 
	Eq.(\ref{m2}) namely
	\begin{equation}\label{a3}
	F(x,t) = 1 + f_1(x,t) + f_2(x,t) + c f_1(x,t) f_2(x,t) 
	\end{equation}
	where $c$ is a parameter. We note again that Eq.(\ref{m2}) contains as particular case the
	relationship used by Hirota in \cite{hirota}. The substitution of Eq.(\ref{a3}) in Eq.(\ref{a2}) leads to
	\begin{eqnarray}\label{a4}
	&& f_{1xxxx}+f_{2xxxx}+3f_{1xx}^2+cf_{1xt} f_2 + f_{1xt} - f_{1t} f_{1x}-f_{1t} f_{2x} -
	\nonumber \\
	&&4 f_{1x} f_{1xxx} -  f_{1x} f_{2t} - 4 f_{1x} f_{2xxx} + 6 f_{1xx} f_{2xx} - 4 f_{1xxx} f_{2x} - f_{2t} f_{2x} - \nonumber \\
	&& 4 f_{2x} f_{2xxx} +  f_2 f_{1xt} + f_2 f_{1xxxx} + f_2 f_{2xt} + f_2 f_{2xxxx} + f_1 f_{1xt} + 
	\nonumber \\
	&& f_1 f_{1xxxx} + f_1 f_{2xt} +
	f_1 f_{2xxxx} + f_{2xt} + c^2 f_2^2 f_1 f_{1xt} + c^2 f_2^2 f_1 f_{1xxxx} - \nonumber \\
	&& c^2 f_2^2  f_{1t} f_{1x} - 4f_2^2 c^2 f_{1x} f_{1xxx} + c^2 f_2 f_1^2 f_{2xt} + c^2 f_2 f_1^2 f_{2xxxx} - \nonumber \\
	&& 12 c^2 f_2 f_{1x}^2 f_{2xx} - c^2 f_1^2 f_{2t} f_{2x} - 
	4 c^2 f_1^2 f_{2x} f_{2xxx} - 12 c^2 f_1 f_{1xx} f_{2x}^2 + \nonumber \\
	&& 2 c f_2 f_1  f_{1xt} + 2 c f_2 f_1  f_{1xxxx} + 2 c f_2 f_1  f_{2xt} + 2 c f_2 f_1  f_{2xxxx} - 2 c f_2  f_{1t} f_{1x} - \nonumber \\
	&& 8 c f_2  f_{1x} f_{1xxx} + 12 c f_2  f_{1xx} f_{2xx} + 12 c f_1  f_{1xx} f_{2xx} - 2 c f_1  f_{2t} f_{2x} - \nonumber \\
	&& 8 c f_1  f_{2x} f_{2xxx}+ 12 c^2 f_{1x}^2 f_{2x}^2 - 12 c f_{1x}^2 f_{2xx} - 12 c f_{1xx} f_{2x}^2 + 3 c^2 f_2^2 f_{1xx}^2 + \nonumber \\
	&& 3 c^2 f_1^2 f_{2xx}^2 + c f_2^2  f_{1xt} + c f_2^2  f_{1xxxx} + 6 c f_2 f_{1xx}^2 + c f_1^2 f_{2xt} + \nonumber \\
	&& c f_1^2 f_{2xxxx} + 6 c f_1 f_{2xx}^2 + c f_1 f_{2xt} + c f_1 f_{2xxxx} + 6 c f_{1xx} f_{2xx} + 4 c f_{1x} f_{2xxx}+ \nonumber \\
	&& c f_{1x} f_{2t} + c f_{1xxxx} f_2 + 12 c^2 f_2 f_1 f_{1xx} f_{2xx} + c f_{1t} f_{2x} + 4 c f_{1xxx} f_{2x} + 3 f_{2xx}^2 =0
	\nonumber \\
	\end{eqnarray}
	\item[3.)] \emph{Equations for the functions $f_1(x,t)$ and $f_2(x,t)$}\\
	The structure of Eq.(\ref{a4}) allow us to assume  a very simple form of the equations for the functions $f_{1,2}$:
	\begin{eqnarray}\label{a5}
	\frac{\partial f_1}{\partial x} &=& \alpha_1 f_1; \ \ \ \frac{\partial f_1}{\partial t} = \beta_1 f_1; \nonumber \\
	\frac{\partial f_2}{\partial x} &=& \alpha_2 f_2; \ \ \ \frac{\partial f_2}{\partial t} = \beta_2 f_2;
	\end{eqnarray}
	This choice will transform Eq.(\ref{a4}) to a polynomial of $f_1$ and $f_2$. Further we assume that
	$\xi = \alpha_1 x + \beta_1 t + \gamma_1$ and $\zeta = \alpha_2 x + \beta_2 t + \gamma_2$ and
	\begin{equation}\label{a6}
	f_1(x,t) = a(\xi); \ \ \ f_2(x,t) = b(\zeta)
	\end{equation}
	Above $\alpha_{1,2}$, $\beta_{1,2}$ and $\gamma_{1,2}$ are parameters.
	\item[4.)] \emph{Relationships connecting  $a(\xi)$ and $b(\zeta)$ to the functions $v(\xi)$ and
		$w(\zeta)$ that are solutions of the simplest equations} \\
	In the discussed here case the relationships are quite simple. We can use Eq.(\ref{i2}) for the cases
	$\mu_1 = \nu_2 = 1$ and $\mu_2 = \nu_4 = 1$. The result is
	\begin{equation}\label{a7}
	a(\xi) = q_1 v(\xi); \ \ \ b(\zeta) = r_1 w(\zeta)
	\end{equation}
	\item[5.)] \emph{Simplest equations for $v(\xi)$ and $w(\zeta)$}\\
	The simplest equations are
	\begin{equation}\label{a8}
	\frac{dv}{d\xi} = v; \ \ \ \frac{dw}{d\zeta} = w
	\end{equation}
	and the corresponding solutions are
	\begin{equation}\label{a9}
	v(\xi) = \omega_1 \exp (\xi); \ \ \ w(\zeta)  = \omega_2 \exp(\zeta)
	\end{equation}
	Below we shall omit the parameters $\omega_{1,2}$ as they can be included in the parameters
	$q_1$ and $r_1$ respectively. We shall omit also $q_1$ and $r_1$ as they can be included in $\xi$ and $\zeta$.
	\item[6.)] \emph{Transformation of Eq.(\ref{a4})}\\
	Let us substitute Eqs.(\ref{a5}) - (\ref{a8}) in Eq.(\ref{a4}). The result is a sum of exponential functions
	and each exponential function is multiplied by a coefficient. Each of these coefficients is a relationship
	containing the parameters of the solution and all of the relationships contain more than one term. Thus
	we don't need to perform a balance procedure.
	\item[7.)] \emph{Obtaining and solving the system of algebraic equations}\\
	The system of algebraic equations is obtained by setting of above-mentioned relationships to $0$.
	Thus we obtain the following system:
	\begin{eqnarray}\label{a10}
	&& \alpha_1^3 + \beta_1 = 0, \nonumber \\
	&& \alpha_2^3 + \beta_2 = 0, \nonumber \\
	&& (c+1) \alpha_1^4 + 4 \alpha_2 (c-1) \alpha_1^3 + 6 \alpha_2^2 (c+1) \alpha_1^2 + [(4c-4)\alpha_2^3 + (\beta_1 + \beta_2) c + \nonumber \\
	&& \beta_1 - \beta_2] \alpha_1 + [(c+1)\alpha_2^3 + (\beta_1 + \beta_2) c - \beta_1 + \beta_2] \alpha_2 = 0.
	\end{eqnarray}
	The non-trivial solution of this system is 
	\begin{eqnarray}\label{a11}
	\beta_1 = -\alpha_1^3; \ \ \beta_2 = -\alpha_2^3; \ \
	c = \frac{(\alpha_1 - \alpha_2)^2}{(\alpha_1 + \alpha_2)^2}
	\end{eqnarray}
	and the corresponding solution of Eq.(\ref{a1}) is
	\begin{eqnarray}\label{a12}
	u(x,t) &=& \frac{12}{\sigma} \frac{\partial^2}{\partial x^2} \Bigg[ 1+ \exp \Big(\alpha_1 x - \alpha_1^3 t + \gamma_1 \Big) + \exp \Big(\alpha_2 x - \alpha_2^3 t + \gamma_2 \Big) + \nonumber \\
	&& \frac{(\alpha_1 - \alpha_2)^2}{(\alpha_1 + \alpha_2)^2}
	\exp \Big( (\alpha_1 + \alpha_2)x - (\alpha_1^3 + \alpha_2^3)t + \gamma_1 + \gamma_2 \Big) \Bigg]
	\end{eqnarray}
	Eq.(\ref{a12}) describes the bisoliton solution of the Korteweg - de Vries equation.
\end{description}
\par 
By this tutorial example we have shown that the SEsM  is capable to search for
multi-soliton solutions of nonlinear PDEs. This capability is acquired on the basis of the 
possibility of use of more than one simplest equation. The relationship (\ref{m2}) 
can be used also for obtaining  exact solution of nonintegrable nonlinear PDEs. This will be 
demonstrated in the following subsection. We note that the particular case of the relationship (\ref{m2}) used by Hirota \cite{hirota} was applied also to the equations of the Korteweg - de Vries hierarchy by Kudryashov \cite{k10}. This fact again shows the potential of the version of the modified method of simplest equation discussed in this text.
\subsection{Application of SEsM for the case of presence of nontrivial transformation of kind (\ref{m1}) to an non-integrable equation - generalized Kawahara equation}
Let us discuss the equation
\begin{equation}\label{b1}
u_t + \left( \sum \limits_{k=0}^l \alpha_k u^k \right) u_x + \beta u_{xxx} + \gamma u^m u_{xxxxx}=0
\end{equation}
where $n$ and $m$ are integers and $\alpha$, $\beta$ and $\gamma$ are parameters.  Solutions of equation of this kind are discussed, e.g., by
Kudryashov \cite{k3} and by Berloff and Howard \cite{bh}
Here we shall apply the version of the modified method of simplest equation from Sect. 2 in presence of a transformation of kind (\ref{m1})
to a particular case of Eq.(\ref{b1}). In Sect. 4 we shall apply discussed version of the modified method 
of simplest equation to two equations from the family of equations (\ref{b1}) for the case
of particular case $u=F$ of the transformation (\ref{m1}).
\par
Let us consider the following particular case of Eq.(\ref{b1}): $\alpha_0 = 0$, $\alpha_1 = 2 \alpha$;
$\alpha_2 =90$, $\alpha_3 =  \dots =0$, $\gamma = -1$, $m=0$. Then from Eq.(\ref{b1}) we obtain
\begin{equation}\label{b2}
u_t + 2 \alpha u u_x + 90 u^2 u_x + \beta u_{xxx} = u_{xxxxx}.
\end{equation}
This equation is known as  a particular case ($\alpha_2 = 90$) of the generalized Kawahara equation \cite{k3}. Let us apply the version of the modified method of simplest equation from Sect.2 for the 
case of transformation obtained on the basis of a Painleve expansion truncated before the "constant term". We note that the goal in this subsection is just to illustrate the application of the discussed version of methodology to a nonitegrable equation. We don't pretend that the
obtained solution (\ref{b13}) is a new one as it is simple and may be obtained by other authors before us. Additional solutions of equations of the kind (\ref{b1}) will be obtained in Sect. 4.  
\par
The 7 steps of the application of the methodology are as follows.
\begin{description}
	\item[1.)] \emph{The transformation}\\
	We shall apply the transformation
	\begin{equation}\label{b2x}
	u(x,t) = 2 [\ln F(x,t)]_{xx}
	\end{equation}
	This transformation can be obtained from the Painleve expansion for Eq.(\ref{b2}) taking into account
	that the leading order is equal to $2$ and truncating the Painleve expansion before the "constant term"
	(i.e. using the expansion $u(x,t) = u_0(x,t)/F^2(x,t) + u_1(x,t)/F(x,t)$). 
	The result of the substitution of Eq.(\ref{b2x}) in Eq.(\ref{b2}) is
	\begin{eqnarray}\label{b3}
	&& -2 F^4 F_{xxxxxxx} + 14 F^3 F_x F_{xxxxxx} + (2 \beta F^4 + 42 F^3 F_{xx} - \nonumber \\
	&& 84 F^2 F_x^2) F_{xxxxx} + (70 F^3 F_{xxx} - 10 F F_x (\beta F^2 + 42 F F_{xx} - \nonumber \\
	&& 42 F_x^2)) F_{xxxx} - 280 F^2 F_x F_{xxx}^2+ (300 F^2 F_{xx}^2 - (- 1080 F_x^2 + \nonumber \\
	&& 4 F^2 (5 \beta-2 \alpha)) F F_{xx} + 8 F_x^2 (-120 F_x^2 + F^2 (5 \beta - \alpha))) F_{xxx} + 
	\nonumber \\
	&& 2 F_{xxt} F^4 - 900 F F_{xx}^3 F_x+ 144 F^2 F_x^3 F_{xx}^2 
	( 5\beta -  2 \alpha)  - \nonumber \\
	&& (40 (3 \beta - \alpha) F_x^3 + 2 F^2 F_t) F F_{xx} - 4 F^3 F_x F_{xt} \nonumber \\
	&& + 16 (3 \beta - \alpha) F_x^5 + 4 F^2 F_t F_x^2 = 0 \nonumber \\
	\end{eqnarray}
	\item[2.)] Relationship among $F(x,t)$ and the functions $f_k(x,t)$ \\
	Here we shall use a very simple particular case of the relationship (\ref{m2}) that contains just 
	one function $f(x,t)$:
	\begin{equation}\label{b4}
	F(x,t) = 1 + f(x,t).
	\end{equation}
	The substitution of Eq.(\ref{b4}) in Eq.(\ref{b3}) leads to
	\begin{eqnarray}\label{b5}
	&&2( \beta f_{xxxxx} + f_{xxt} - f_{xxxxxxx}) f^4 + (2(-5 \beta f_{xxxx} - 2 f_{xt} + 7 f_{xxxxxx}) f_x +
	\nonumber \\
	&&((-20 \beta + 8 \alpha) f_{xxx} + 42 f_{xxxxx} - 2 f_t) f_{xx} + 70 f_{xxx} f_{xxxx} + 8 \beta f_{xxxxx} - \nonumber \\
	&& 8 f_{xxxxxxx} + 8 f_{xxt}) f^3 + (((40 \beta - 8 \alpha) f_{xxx} - 84 f_{xxxxx} +4 f_t) f_x^2 +
	\nonumber \\
	&& ((60 \beta - 24 \alpha) f_{xx}^2 - 420 f_{xxxx} f_{xx} - 30 f_{xxxx} \beta - 280 f_{xxx}^2 + 42 f_{xxxxxx} - 
	\nonumber \\
	&& 12 f_{xt}) f_x + 300 f_{xx}^2 f_{xxx} + ((- 60 \beta + 24 \alpha) f_{xxx} + 126 f_{xxxxx} - 6 f_t) 
	f_{xx}+ 
	\nonumber \\
	&&210 f_{xxx} f_{xxxx} + 12 \beta f_{xxxxx} - 12 f_{xxxxxxx} + 12 f_{xxt}) f^2 + (((- 120 \beta+ 
	\nonumber \\
	&& 40 \alpha) f_{xx} +  420 f_{xxxx}) f_x^3 + (1080 f_{xx}  f_{xxx} + (80 \beta - 16 \alpha) f_{xxx} - 168 f_{xxxxx}+ \nonumber \\
	&& 8 f_t) f_x^2 + (-900 f_{xx}^3 + (120 \beta - 48 \alpha) f_{xx}^2 - 840 f_{xxxx}f_{xx} - 30 f_{xxxx}\beta - 560 f_{xxx}^2 + \nonumber \\
	&& 42 f_{xxxxxx} - 12 f_{xt}) f_x + 600 f_{xx}^2 f_{xxx} + ((-60 \beta + 24 \alpha) f_{xxx} + 126 f_{xxxxx} - \nonumber \\
	&& 6 f_t) f_{xx} + 210 f_{xxx}  f_{xxxx} + 8 \beta f_{xxxxx} - 8 f_{xxxxxxx} + 8 f_{xxt}) f+ (48 \beta - \nonumber \\
	&& 16 \alpha) f_x^5 - 960 f_x^4 f_{xxx} + (720 f_{xx}^2 + (-120 \beta + 40 \alpha) f_{xx} + 420 f_{xxxx}) f_x^3 + \nonumber \\
	&& (1080 f_{xx} f_{xxx} + (40 \beta - 8 \alpha) f_{xxx} - 84 f_{xxxxx} + 4 f_t) f_x^2 + (-900 f_{xx}^3 +
	\nonumber \\
	&& (60 \beta - 24 \alpha) f_{xx}^2 - 420 f_{xxxx} f_{xx} - 10 \beta f_{xxxx} - 280 f_{xxx}^2 + 14 f_{xxxxxx} - 4 f_{xt}) f_x + \nonumber \\
	&& 300 f_{xx}^2 f_{xxx} + ((-20 \beta + 8 \alpha) f_{xxx} + 42 f_{xxxxx} - 2 f_t) f_{xx} + 70 f_{xxx} f_{xxxx} + \nonumber \\
	&& 2 \beta f_{xxxxx} - 2 f_{xxxxxxx} + 2 f_{xxt} = 0
	\nonumber \\
	\end{eqnarray}
	\item[3.)] \emph{Equation for the function $f(x,t)$}\\
	The structure of Eq.(\ref{b5}) allows us to use the following equations for $f(x,t)$
	\begin{equation}\label{b6}
	\frac{\partial f}{\partial x} = k f; \ \ \ \frac{\partial f}{\partial t} = \omega f
	\end{equation}
	This choice will transform Eq.(\ref{b5}) to a polynomial of $f(x,t)$. Further we assume that
	$\xi = kx + \omega t + \sigma$ and
	\begin{equation}\label{b7}
	f(x,t) = a (\xi)
	\end{equation}
	\item[4.)] \emph{Relationship between $a(\xi)$ and a function that is solution of a simplest equation}\\
	The relationship is very simple. We can use Eq.(\ref{i2}) for the cases
	$\mu_1 = \nu_2 = 1$. The result is
	\begin{equation}\label{b8}
	a(\xi) = q v(\xi)
	\end{equation}
	where $q$ is a parameter.
	\item[5.)] \emph{The simplest equations for $v(\xi)$}\\
	The simplest equations is
	\begin{equation}\label{b9}
	\frac{dv}{d\xi} = v
	\end{equation}
	and the corresponding solutions are
	\begin{equation}\label{b10}
	v(\xi) = \theta \exp (\xi)
	\end{equation}
	Below we shall omit the parameter $\theta$ as it can be included in the parameter
	$q$. We shall omit also $q$  as it can be included in $\xi$.
	\item[6.)] \emph{Transformation of Eq.(\ref{b5})}\\
	Let us substitute Eqs.(\ref{b6}) - (\ref{b10}) in Eq.(\ref{b5}). The result is a sum of exponential 
	functions and each exponential function is multiplied by a coefficient. Any of these coefficients is a 
	relationship containing  parameters of the solution and parameters of the solved equation. 
	All of these relationships contain more than one term. Thus we don't need to perform a balance procedure.
	\item[7.)] \emph{Obtaining and solving the system of algebraic equations}\\
	The system of algebraic equations is obtained by setting of above-mentioned coefficients to $0$.
	Thus we obtain the following system of algebraic equations
	\begin{eqnarray}\label{b11}
	&& 57 k^5  - (9 \beta  - 4 \alpha) k^3 + 3\omega = 0 \nonumber \\
	&& 29 k^5 - (5 \beta  - 2 \alpha) k^3 +  \omega =0 \nonumber \\
	&& k^5 - \beta k^3 - \omega = 0
	\end{eqnarray}
	One nontrivial solution of this system is
	\begin{equation}\label{b12}
	k = \sqrt{\frac{3 \beta - \alpha}{15}}; \ \ \ \omega = \frac{(\alpha - 3 \beta)(\alpha + 12 \beta)
		\sqrt{15(3 \beta - \alpha)}}{3375}
	\end{equation}
	and the corresponding solution of Eq.(\ref{b2}) becomes
	\begin{equation}\label{b13}
	u(x,t) = \frac{2}{15} (\alpha - 3 \beta) \frac{\exp \left \{ \frac{4 \sqrt{15}}{375} \sqrt{3 \beta - \alpha} \left[ \frac{25}{4} x - \left(\beta - \frac{\alpha}{3} \right) \left(\beta + \frac{\alpha}{12} \right) t  \right]\right \} }{\left \{ 1 +  \exp \left \{ \frac{4 \sqrt{15}}{375} \sqrt{3 \beta - \alpha} \left[ \frac{25}{4} x - \left(\beta - \frac{\alpha}{3} \right) \left(\beta + \frac{\alpha}{12} \right) t  \right]\right \} \right \}^2}
	\end{equation}
	This solution is shown in Fig. 2.
	\begin{figure}[!htb]
		\centering
		\includegraphics[scale=0.8]{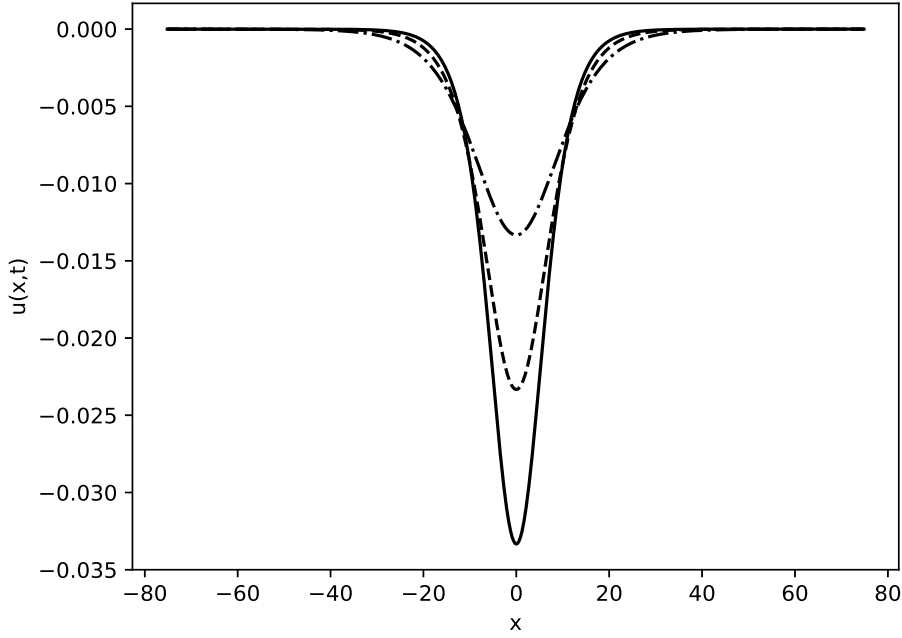}
		\caption{Several profiles of the solutions of Eq.(\ref{b13}). The values of parameters
			are as follows. Solid line: $\alpha=0.5$, $\beta = 0.5$; dashed line: $\alpha=0.5$, 
			$\beta = 0.4$; dot-dashed line: $\alpha=0.5$, $\beta = 0.3$ .}
	\end{figure}
\end{description}
\section{Obtaining solutions of Eq.(\ref{b1}) by use of the particular case $u(x,t)=F(x,t)$ of 
	transformation (\ref{m1})}
Let us now consider Eq.(\ref{b1}) and apply the methodology for the particular case $u(x,t)=F(x,t)$ of
transformation (\ref{m1}). The steps of the methodology are as follows
\begin{description}
	\item[1.)] \emph{The transformation}\\
	We shall  use a  particular case of transformation (\ref{m1}), i.e. $u(x,t)=F(x,t)$.
	\item[2.)] \emph{Relationship among $F(x,t)$ and the functions $f_k(x,t)$}\\
	The function F(x,t) will be searched as a function of another function $f(x,t)$ and the corresponding relationship 
	is particular case of the relationship (\ref{m2})
	\begin{equation}\label{c1}
	F(x,t) = \sum \limits_{i=0}^N \gamma_i f(x,t)^i
	\end{equation}
	where $\gamma_i$ are parameters.
	\item[3.)] \emph{Equation for the function $f(x,t)$}\\
	The function $f(x,t)$ will be assumed to be a traveling wave
	\begin{equation}\label{c2}
	f(x,t) = a(\xi); \ \ \ \xi = \mu x + \nu t
	\end{equation}
	\item[4.)] \emph{Representation of the function $a(\xi)$ by a function that is solution of a simplest equation}\\
	We shall not express further the function $a(\xi)$ through another function $v(\xi)$ and
	instead of this we shall assume that $a(\xi)$ is a solution of a simplest equation of the class
	(\ref{i1}). 
	\item[5.)] \emph{The simplest equations}\\
	Below we shall use the following two simplest equations:
	\begin{equation}\label{c3}
	\frac{da}{d\xi}  = \sum \limits_{j=0}^{p} d_j a^j,
	\end{equation}
	and
	\begin{equation}\label{c4}
	\left (\frac{d a}{d\xi} \right)^2 = \sum \limits_{j=0}^{p} d_j a^j.
	\end{equation}
	First we shall use Eq.(\ref{c3}) as simplest equation. Then we shall return to step 5.) and we shall
	use Eq.(\ref{c4}) as a simplest equation.
	\item[6.)] \emph{Transformation of Eq.(\ref{b1})}\\
	The substitution of Eqs. (\ref{c1}) and (\ref{c3}) in Eq.(\ref{b1})
	leads to a polynomial of $a(\xi)$ that contains the following maximum powers of the terms of Eq.(\ref{b1}) : $N+p-1$; $N+3(p-1)$; $Nm+n+5(p-1)$; $Nl+N+p-1$. 
	In order to obtain the system of nonlinear algebraic
	equations we have to write balance equations for these powers, i.e. in this case we have to balance 
	the largest powers:  
	$Nm+n+5(p-1)$ and $Nl+N+p-1$. This leads us to the balance equation
	\begin{equation}\label{c5}
	N(l-m) = 4(p-1)
	\end{equation}
	We note that $l$, $m$, $p$, $N$ have to be integers or $0$. We have $p>1$ and $l>m$. Then from Eq.(\ref{c5})
	\begin{equation}\label{c6}
	N = 4 \frac{p-1}{l-m}
	\end{equation}
	which means that  the equations of the class Eq.(\ref{b1}) may have solutions of the kind
	\begin{equation}\label{c7}
	u(x,t) = \sum \limits_{i=0}^{4 \frac{p-1}{l-m}} \gamma_i a(\xi)^i,
	\end{equation}
	where $\xi = \mu x + \nu t$ and $a(\xi)$ is a solution of the simplest equation
	\begin{equation}\label{c8}
	\frac{da}{d\xi} = d_0 + \dots + d_p a^p.
	\end{equation}
	We note that $\frac{p-1}{l-m}$ must be an integer. The solution of Eq.(\ref{b2}) is particular case of the 
	the solution of the above class of equations when $l=1$, $m=0$, $\alpha_0=0$, $\alpha_1 = 2 \alpha$,
	$\alpha_2 = 90$, $\alpha_3 = \dots = 0$, $\gamma=-1$. In this case 
	$N=4(p-1)$.
	\item[7.)] \emph{Systems of nonlinear algebraic equations and their solutions}\\
	Let us discuss the case $l=1$. Then $m=0$. Let in addition $p=2$. Thus $N=4$. We shall solve the equation
	\begin{equation}\label{c8}
	u_t + \left( \alpha_0 + \alpha_1 u \right) u_x + \beta u_{xxx} + \gamma u_{xxxxx}=0
	\end{equation}
	where the solution $u(x,t) = u(\xi)$, $\xi = \mu x + \nu t$ will be searched in the form
	\begin{equation}\label{c9}
	u(\xi) = \gamma_0 + \gamma_1 a(\xi) + \gamma_2 a(\xi)^2 + \gamma_3 a(\xi)^3 + \gamma_4 a(\xi)^4
	\end{equation}
	where $a(\xi)$ is solution of the simplest equation
	\begin{equation}\label{c10}
	\frac{da}{a \xi} = d_0 + d_1 a(\xi) + d_2 a(\xi)^2
	\end{equation}
	Above $\mu$, $\nu$, $\gamma_{0,1,2,3,4}$ and $d_{0,1,2}$ are parameters.
	We note that Eq.(\ref{c10}) is the Riccati equation. Below we shall use the following solution
	of this equation
	\begin{equation}\label{sol_riccati}
	a(\xi) = - \frac{d_1}{2d_2} - \frac{\theta}{2 d_2} \tanh \Bigg[\frac{\theta (\xi + \xi_0)}{2} \Bigg]
	\end{equation}
	where $\xi_0$ is constant of integration and $\theta^2 = d_1^2 - 4 d_0d_2 > 0$.
	\par
	The substitution of Eqs. (\ref{c9}) and (\ref{c10}) in Eq.(\ref{c8}) leads to a system of 10
	nonlinear algebraic equations. We shall not write this system here as it is quite long. The system can 
	be solved by means of a computer algebra software. One nontrivial solution is
	\begin{eqnarray}\label{c11}
	d_0 &=& \frac{\beta + 13 d_1^2 \gamma \mu^2}{52 d_2 \gamma \mu^2}; \ \ \gamma_4 = -1680 \frac{d_2^4 \gamma \mu^4}{\alpha_1}
	\nonumber \\
	\gamma_3 &=& -3360 \frac{\mu^4 d_2^3 \gamma d_1}{\alpha_1}; \ \ \gamma_2 = -\frac{840}{13} \frac{\mu^2 d_2^2 (\beta + 
		39 d_1^2 \gamma \mu^2)}{\alpha_1} \nonumber \\
	\gamma_1&=& -\frac{840}{13} \frac{\mu^2 d_1 d_2 (\beta + 13 d_1^2 \gamma \mu^2)}{\alpha_1} \nonumber \\
	\gamma_0&=& -\frac{17745 d_1^4 \gamma^2 \mu^5 + \gamma (2730 \beta d_1^2 \mu^3 + 169 \alpha_0 \mu + 169 \nu)  + 
		69 \beta^2 \mu}{169 \mu \alpha_1 \gamma} \nonumber \\
	\end{eqnarray}
	On the basis of Eqs.(\ref{c11}) and (\ref{sol_riccati}) we obtain, e.g.,
	\begin{eqnarray}\label{c12}
	\mu^2 &>& \frac{\beta}{13(\gamma - 1)}; \ \ \theta^2 = d_1^2 \Big(1- \frac{\beta + 13 \mu^2}{13 \gamma \mu^2} \Big) ;
	\nonumber \\
	a(\xi) &=&  - \frac{d_1}{2d_2} \Bigg \{ 1+ \sqrt{ \Big(1- \frac{\beta + 13 \mu^2}{13 \gamma \mu^2} \Big)} \tanh 
	\Bigg[\frac{d_1 \sqrt{ \Big(1- \frac{\beta + 13 \mu^2}{13 \gamma \mu^2} \Big)} (\xi + \xi_0)}{2} \Bigg] \Bigg \}
	\nonumber \\
	\end{eqnarray}
	Then the solution of Eq.(\ref{c9}) is 
	\begin{eqnarray}\label{c13}
	u(x,t) = u(\xi) = -\frac{17745 d_1^4 \gamma^2 \mu^5 + \gamma (2730 \beta d_1^2 \mu^3 + 169 \alpha_0 \mu + 169 \nu)  + 
		69 \beta^2 \mu}{169 \mu \alpha_1 \gamma} + \nonumber \\
	\frac{420}{13} \frac{\mu^2 d_1^2  (\beta + 13 d_1^2 \gamma \mu^2)}{\alpha_1}\Bigg \{ 1+ \sqrt{ \Big(1- \frac{\beta + 13 \mu^2}{13 \gamma \mu^2} \Big)} \tanh 
	\Bigg[\frac{d_1 \sqrt{ \Big(1- \frac{\beta + 13 \mu^2}{13 \gamma \mu^2} \Big)} (\xi + \xi_0)}{2} \Bigg] \Bigg \} -
	\nonumber \\
	\frac{210}{13} \frac{\mu^2 d_1^2 (\beta + 39 d_1^2 \gamma \mu^2)}{\alpha_1} \Bigg \{ 1+ \sqrt{ \Big(1- \frac{\beta + 
			13 \mu^2}{13 \gamma \mu^2} \Big)} \tanh 
	\Bigg[\frac{d_1 \sqrt{ \Big(1- \frac{\beta + 13 \mu^2}{13 \gamma \mu^2} \Big)} (\xi + \xi_0)}{2} \Bigg] \Bigg \}^2 +
	\nonumber \\
	420 \frac{\mu^4 d_1^4 \gamma}{\alpha_1}\Bigg \{ 1+ \sqrt{ \Big(1- \frac{\beta + 13 \mu^2}{13 \gamma \mu^2} \Big)} \tanh 
	\Bigg[\frac{d_1 \sqrt{ \Big(1- \frac{\beta + 13 \mu^2}{13 \gamma \mu^2} \Big)} (\xi + \xi_0)}{2} \Bigg] \Bigg \}^3 -
	\nonumber \\
	-105 \frac{d_1^4 \gamma \mu^4}{\alpha_1} \Bigg \{ 1+ \sqrt{ \Big(1- \frac{\beta + 13 \mu^2}{13 \gamma \mu^2} \Big)} \tanh 
	\Bigg[\frac{d_1 \sqrt{ \Big(1- \frac{\beta + 13 \mu^2}{13 \gamma \mu^2} \Big)} (\xi + \xi_0)}{2} \Bigg] \Bigg \}^4
	\nonumber \\
	\end{eqnarray}
	This solution is shown in Fig.3.
	\begin{figure}[!htb]
		\centering
		\includegraphics[scale=0.8]{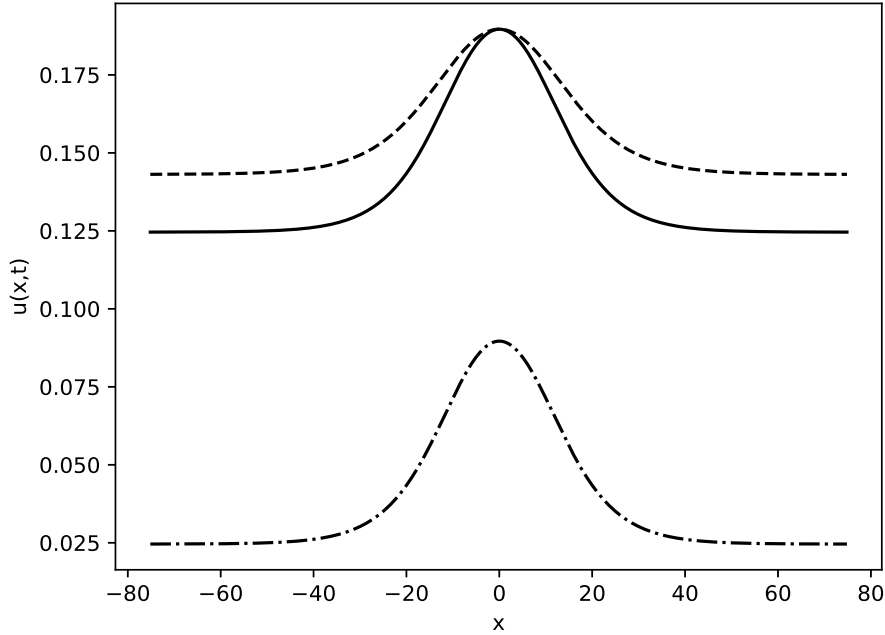}
		\caption{Several profiles of the solutions of Eq.(\ref{c13}). The values of parameters
			are: $\gamma=2$; $\beta=0.7$;  $\alpha_1=5$, $\nu=0.15$, $\mu = 0.1+(\frac{\beta}{13(\gamma-1)})^{1/2}$. Solid line: $d_1=0.8$, $\alpha_0 = -1.5$; dashed line:
			$d_1=0.7$, $\alpha_0=-1.5$; dot-dashed line: $d_1=$0.8, $\alpha_0=-1$.}
	\end{figure}
	\par
	We can continue to search for other solutions of Eq.(\ref{c8}). For an example if we set $p=3$ then $N=8$.
	This means that the simplest equation will be
	\begin{equation}\label{c14}
	\frac{da}{a \xi} = d_0 + d_1 a(\xi) + d_2 a(\xi)^2 + d_3 a(\xi)^3
	\end{equation}
	and the solution of Eq.(\ref{c8}) will be of the kind
	\begin{eqnarray}\label{c15}
	u(\xi) &=& \gamma_0 + \gamma_1 a(\xi) + \gamma_2 a(\xi)^2 + \gamma_3 a(\xi)^3 + \gamma_4 a(\xi)^4 +
	\gamma_5 a(\xi)^5 + \nonumber \\
	&& \gamma_6 a(\xi)^6 + \gamma_7 a(\xi)^7 + \gamma_8 a(\xi)^8
	\end{eqnarray}
	We note that Eq.(\ref{c14}) is a particular case of the Abel equation of the first kind.
	We let the obtaining of this solution to the interested reader and
	let us now return to the Step 5.) and use Eq. (\ref{c4}) as a simplest equation. The next step
	is
	\item [6.)] \emph{Transformation of Eq.(\ref{b1})}\\
	The use of Eq.(\ref{c4}) as simplest equation will change the balance equation.
	Now we have to balance the following powers: $N-1$, $Nl+N-1$, $N+p-3$, $Nm + N+2p-5$.
	We have to consider the cases $m=0$ and $m>0$. For the case $m=0$ when $p=1$ then $N+p-3 > Nm+N+2p-5$.
	The balance could be $N+p-3 = N-2 = Nl + N-1$. But this leads to $Nl=-1$ and such a balance is impossible. In other words there is no balance for the case $m=0$, $p=0$. For the case $m=0$ and 
	$p \ge 1$ $N+p-3 \le Nm+N+2p-5$ and the balance will be (we keep $m$ in the balance equation
	despite $m=0$ because this balance will be valid also for other cases)
	\begin{equation}\label{c16}
	N = \frac{2(p-2)}{l-m}
	\end{equation}
	This balance is different with respect to the balance given by Eq.(\ref{c6}). Let now $m \ge 1$.
	We have again to balance the terms $Nl+N-1$ and $Nm + N+2p-5$ that leads again to the balance 
	equation (\ref{c16}).
	\par
	All above means that the equations of the class (\ref{b1}) may have solutions of the kind
	\begin{equation}\label{c17}
	u(x,t) = u(\xi) = \sum \limits_{i=0}^{\frac{2(p-2)}{l-m}} \gamma_i a(\xi)^i,
	\end{equation}
	where $\xi = \mu x + \nu t$ and
	\begin{equation}\label{c18}
	\left( \frac{da}{d\xi}\right)^2 = d_0 + \dots+ d_p a^p.
	\end{equation}
	\item[7.)] \emph{Systems of nonlinear algebraic equations and their solutions}\\
	Let us now solve two nonlinear algebraic systems and obtain some exact solutions on the basis 
	of the simplest equation (\ref{c18}). First of all we shall consider the case $l=1$, $m=0$. In this case 
	the balance equation is $N=2(p-2)$ and the smallest possible value of $p$ is $p=3$. Then $N=2$. Thus the simplest equation becomes
	\begin{equation}\label{c19}
	\left( \frac{da}{d\xi}\right)^2 = d_0 + d_1 a + d_2 a^2 + d_3 a^3,
	\end{equation}
	and the solution of Eq.(\ref{b1}) will be given by
	\begin{equation}\label{c20}
	u(\xi) = \gamma_0 + \gamma_1 a(\xi) + \gamma_2 a(\xi)^2.
	\end{equation}
	The general solution of Eq.(\ref{c19}) is given by the special function $V$ discussed in \cite{vdv15}.
	Below we shall use the particular case of Eq.(\ref{c19}) where $d_2=0$, $d_3 = 4$, $d_0 = - g_3$,
	$d_1 = - g_2$. This  particular case of Eq.(\ref{c19}) is the differential equation for the elliptic function
	of Weierstrass that we shall denote as $\wp(\xi;g_1,g_2,g_3)$.
	\par 
	The substitution of Eqs(\ref{c20}) and the equation for the elliptic function of Weierstrass that is 
	particular case of Eq. (\ref{c19}) transforms Eq.(\ref{b1}) to a polynomial of $a(\xi)$. We set to 0 
	the coefficients of this polynomial and obtain the following system of 4 nonlinear algebraic
	equations
	\begin{eqnarray}\label{c21}
	&&1680 \gamma  \mu^4 +  \alpha_1 \gamma_2 = 0 \nonumber \\
	&&120 \gamma \gamma_1 \mu^4 + 20 \beta \gamma_2 \mu^2 +  \alpha_1 \gamma_1 \gamma_2  = 0 \nonumber \\
	&&-336 \mu^5 \gamma_2 g_2 \gamma + 12 \mu^3 \beta \gamma_1 + \mu [(2 \alpha_1 \gamma_0 + 2 \alpha_0) \gamma_2 + \alpha_1 \gamma_1^2] + 2 \gamma_2 \nu = 0 \nonumber \\
	&&(-18 g_2 \gamma \gamma_1 - 120 g_3 \gamma \gamma_2) \mu^5 - 3 \beta g_2 \gamma_2 \mu^3 + \gamma_1 (\alpha_1 \gamma_0 + \alpha_0) \mu + \gamma_1 \nu = 0 \nonumber \\
	\end{eqnarray}
	One non-trivial solution of the system (\ref{c21}) is
	\begin{eqnarray}\label{c22}
	g_2 &=& \frac{ 31 \beta^3 - 4745520 g_3 \gamma^3 \mu^6}{42588 \mu^4 \gamma^2 \beta}; \ \ 
	\gamma_2 = -1680 \frac{\gamma \mu^4}{\alpha_1}; \ \  
	\gamma_1 = -280 \frac{\beta \mu^2}{13 \alpha_1} \nonumber \\ 
	\gamma_0 &=& \frac{31 \beta^3 \mu - 169 \beta \gamma (\alpha_0 \mu + \nu)  - 3163680 g_3 \gamma^3 \mu^7}{169 \mu \beta \alpha_1 \gamma}
	\end{eqnarray}
	This solution leads to the following solution of Eq.(\ref{b1}) (note that $l=1$ and $m=0$)
	\begin{eqnarray}\label{c23}
	u(x,t) &=& u(\xi) = \frac{31 \beta^3 \mu - 169 \beta \gamma (\alpha_0 \mu + \nu)  - 3163680 g_3 \gamma^3 \mu^7}{169 \mu \beta \alpha_1 \gamma} -\nonumber \\
	&& 280 \frac{\beta \mu^2}{13 \alpha_1} \wp \Bigg(\xi;g_1,\frac{ 31 \beta^3 - 4745520 g_3 \gamma^3 \mu^6}{42588 \mu^4 \gamma^2 \beta},g_3\Bigg) - \nonumber \\
	&& -1680 \frac{\gamma \mu^4}{\alpha_1}\wp \Bigg(\xi;g_1,\frac{ 31 \beta^3 - 4745520 g_3 \gamma^3 \mu^6}{42588 \mu^4 \gamma^2 \beta},g_3\Bigg)^2
	\end{eqnarray}
	\par 
	Let us now consider again the case $l=1$, $m=0$ but now we set $p=4$. In this case $N=4$.
	Thus the simplest equation becomes
	\begin{equation}\label{c24}
	\left( \frac{da}{d\xi}\right)^2 = d_0 + d_1 a + d_2 a^2 + d_3 a^3 + d_4 a^4,
	\end{equation}
	and the solution of Eq.(\ref{b1}) will be given by
	\begin{equation}\label{c25}
	u(\xi) = \gamma_0 + \gamma_1 a(\xi) + \gamma_2 a(\xi)^2 + \gamma_3 a(\xi)^3 + \gamma_4 a(\xi)^4.
	\end{equation}
	The general solution of Eq.(\ref{c24}) is given by the special function $V$ discussed in \cite{vdv15}.
	Below we shall use the particular case of Eq.(\ref{c24}) where $d_1=0$, $d_3 = 0$. This  particular 
	case of Eq.(\ref{c24}) is the differential equation for the elliptic functions of Jacobi.
	\par 
	The substitution of Eqs. (\ref{c24}) and (\ref{c25}) in Eq.(\ref{b1}) transforms it to a polynomial
	of $a(\xi)$. We set to 0 
	the coefficients of this polynomial and obtain a system of 8 nonlinear algebraic
	equations. As the algebraic system is relatively long we shall not write it here. One nontrivial solution
	of this algebraic system is
	\begin{eqnarray}\label{c26}
	\gamma_0 &=& \frac{1}{43940 \alpha_1 d_2 \gamma^2 \mu^3 - 1183 \alpha_1 \beta \gamma \mu}\bigg[ -4921280 d_2^3 \gamma^3 \mu^7 - 94640 d_2 \bigg( d_2 \beta \mu^3 + \nonumber \\
	&& \frac{13}{28} \mu \alpha_0 + \frac{13}{28} \nu \bigg) \mu^2 \gamma^2 + 11180 \beta \bigg( d_2 \beta \mu^3 + \frac{91}{860} \mu \alpha_0 + \frac{91}{860} \nu \bigg) \gamma - 217 \beta^3 \mu \bigg] \nonumber \\
	\gamma_1 &=& 0; \ \ \gamma_2 = -\frac{280}{13} \frac{\mu^2 d_4 (52 d_2 \gamma \mu^2 + \beta)}{\alpha_1};
	\ \ \gamma_3 = 0; \ \ \gamma_4 = -1680 \frac{d_4^2 \gamma \mu^4}{\alpha_1} \nonumber \\
	d_0&=& \frac{1406080 d_2^3 \gamma^3 \mu^6 - 56784 \beta d_2^2 \gamma^2 \mu^4 + 31 \beta^3}{6327360 \gamma^2 (d_2 \gamma \mu^2 - \frac{7}{260} \beta) d_4 \mu^4}
	\end{eqnarray}
	Thus the simplest equation becomes
	\begin{equation}\label{c27}
	\left( \frac{da}{d\xi}\right)^2 = \frac{1406080 d_2^3 \gamma^3 \mu^6 - 56784 \beta d_2^2 \gamma^2 \mu^4 + 31 \beta^3}{6327360 \gamma^2 (d_2 \gamma \mu^2 - \frac{7}{260} \beta) d_4 \mu^4} + d_2 a^2  + d_4 a^4
	\end{equation}
	For some selected values of the parameters Eq.(\ref{c27}) has the Jacobi elliptic functions as solutions.
	Let us consider an  example.
	Let 
	\begin{eqnarray*}
		k^2 = 1 - \frac{1406080 d_2^3 \gamma^3 \mu^6 - 56784 \beta d_2^2 \gamma^2 \mu^4 + 31 \beta^3}{6327360 \gamma^2 (d_2 \gamma \mu^2 - \frac{7}{260} \beta) d_4 \mu^4},
	\end{eqnarray*}
	and $d_2$ and $d_4$ are solutions of the system
	\begin{eqnarray*}
		d_2 = 1 - 2  \frac{1406080 d_2^3 \gamma^3 \mu^6 - 56784 \beta d_2^2 \gamma^2 \mu^4 + 31 \beta^3}{6327360 \gamma^2 (d_2 \gamma \mu^2 - \frac{7}{260} \beta) d_4 \mu^4}; \nonumber \\
		d_4 = - 1 + \frac{1406080 d_2^3 \gamma^3 \mu^6 - 56784 \beta d_2^2 \gamma^2 \mu^4 + 31 \beta^3}{6327360 \gamma^2 (d_2 \gamma \mu^2 - \frac{7}{260} \beta) d_4 \mu^4}
	\end{eqnarray*}
	Then the Jacobi elliptic function ${\rm cn}(\xi;k)$ is solution of Eq.(\ref{c27}). The solution of Eq.(\ref{b1}) becomes
	\begin{eqnarray}\label{c28}
	u(\xi) &=& \frac{1}{43940 \alpha_1 d_2 \gamma^2 \mu^3 - 1183 \alpha_1 \beta \gamma \mu}\bigg[ -4921280 d_2^3 \gamma^3 \mu^7 - 94640 d_2 \bigg( d_2 \beta \mu^3 + \nonumber \\
	&& \frac{13}{28} \mu \alpha_0 + \frac{13}{28} \nu \bigg) \mu^2 \gamma^2 + 11180 \beta \bigg( d_2 \beta \mu^3 + \frac{91}{860} \mu \alpha_0 + \frac{91}{860} \nu \bigg) \gamma - 217 \beta^3 \mu \bigg]  -
	\nonumber \\
	&& \frac{280}{13} \frac{\mu^2 d_4 (52 d_2 \gamma \mu^2 + \beta)}{\alpha_1} {\rm cn}(\xi;k)^2  -1680 \frac{d_4^2 \gamma \mu^4}{\alpha_1}  {\rm cn}(\xi;k)^4.
	\end{eqnarray}
	If $k =1$ then ${\rm cn}(\xi;1) = \frac{1}{\cosh(\xi)}$. This happens when $\mu = \pm \frac{1}{2} \left( \frac{\beta}{13 \gamma}\right)^{1/2}$. In addition $d_2=1$ and $d_4=-1$. Then the solution (\ref{c28}) can be expressed by elementary functions
	\begin{eqnarray}\label{c29}
	u(\xi) &=&   \frac{1}{43940 \alpha_1  \gamma^2 \mu^3 - 1183 \alpha_1 \beta \gamma \mu}\bigg[ -4921280  \gamma^3 \mu^7 - 94640  \bigg( \beta \mu^3 + \nonumber \\
	&& \frac{13}{28} \mu \alpha_0 + \frac{13}{28} \nu \bigg) \mu^2 \gamma^2 + 11180 \beta \bigg(  \beta \mu^3 + \frac{91}{860} \mu \alpha_0 + \frac{91}{860} \nu \bigg) \gamma - 217 \beta^3 \mu \bigg] 
	+ \nonumber \\
	&& \frac{280}{13} \frac{\mu^2  (52  \gamma \mu^2 + \beta)}{\alpha_1 \cosh^2(\xi)}  -1680 \frac{ \gamma \mu^4}{\alpha_1 \cosh^4(\xi)}  
	\end{eqnarray}
	where $\xi = \pm \frac{1}{2} \left( \frac{\beta}{13 \gamma}\right)^{1/2} x + \nu t$.
\end{description}
The profile of (\ref{c29}) is shown in Fig.4. for several values of the parameters of the solution.
\begin{figure}[!htb]
	\centering
	\includegraphics[scale=0.8]{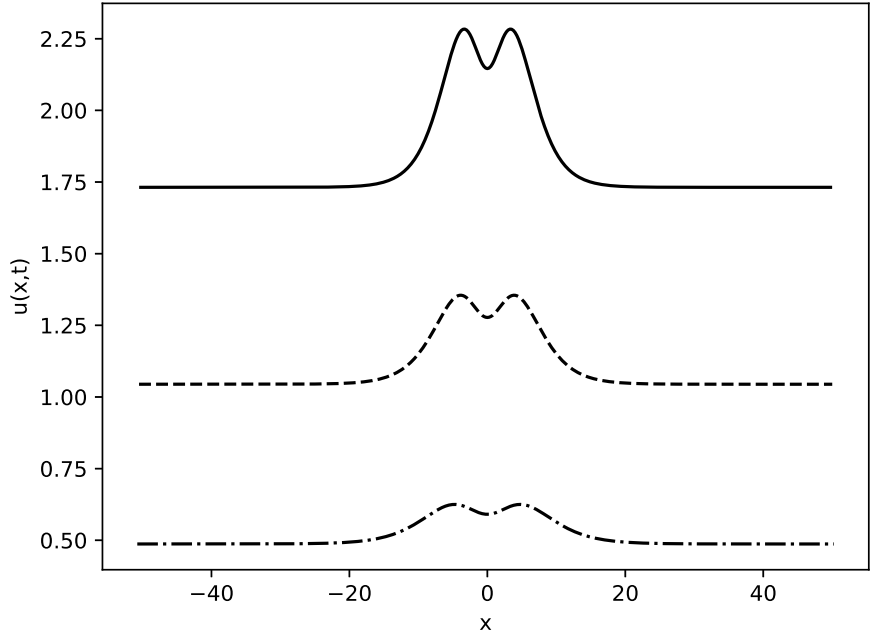}
	\caption{Several profiles of the solutions of Eq.(\ref{c29}). The values of parameters
		are: $\alpha_1=4$; $\alpha_0=0.1$;  $\gamma=2.0$, $\nu=0.15$. Solid line: $\beta = 4$; dashed line:
		$\beta=3$; dot-dashed line: $\beta=2$.}
\end{figure}
\section{Another example connected to simplest equation with fractional powers and more than one balance equation}
\subsection{Solution for the case of integer values of powers $a$ and $b$ in Eq.(\ref{ex1})}
Let us now consider the equation
\begin{equation}\label{ex1}
u_t = p(u^a)_x + q(u^b)_{xx}
\end{equation}
which is a kind of  a nonlinear reaction-diffusion equation \cite{gilding} and has application in
the  modeling of flow motion a in porous medium. First we shall assume that $p$, $q$, $a$, and 
$b$ are integers. The application of the SEsM is as follows
\begin{description}
	\item[1.)] \emph{The transformation}\\
	We shall use a particular case of transformation (\ref{m1}), i.e., $u(x,t)=F(x,t)$.
	\item[2.)] \emph{Relationship among $F(x,t)$ and functions $f_k(x,t)$} \\
	We shall use a single function $f(x,t)$ and the following particular case of Eq.(\ref{m2})
	\begin{equation}\label{ey1}
	u(x,t) = \sum \limits_{i=0}^N \gamma_i f(x,y)^i.
	\end{equation}
	\item[3.)] \emph{Expression of the function $f(x,t)$ by another functions that are connected to solutions of simplest equations} \\
	We shall use the most simple expression containing just one function that is a traveling wave: 
	$f(x,t) = e(\xi)$;  $\xi = \alpha x + \beta t$
	\item[4.)] \emph{Expression of $e(\xi)$ by function that is  a solution of a simplest equation}\\
	We shall use again the most simple relationship $e(\xi) = c(\xi)$
	\item[5.)] \emph{The simplest equation} \\
	Here we shall use the following simplest equation
	\begin{equation}\label{ey2}
	\frac{dc}{d\xi} = \sum \limits_{j=0}^r d_j c^j,
	\end{equation}
	where $r$ and $d_j$, $j=0,\dots$ are parameters.
	\item[6.)] \emph{Transformation of Eq.(\ref{ex1})} \\
	The application of the balance procedure to Eq.(\ref{ex1}) leads to the balance equation
	\begin{equation}\label{ba1}
	N = \frac{r-1}{a-b}
	\end{equation}
	Let us consider the case $r=3$. Then $a=b+2$ and $N=1$ or $a=b+1$ and $N=2$. Let us consider the 
	last case and set $a=3$ and $b=2$. Thus we shall search a solution of 
	\begin{equation}\label{ex1z}
	u_t = p(u^3)_x + q(u^2)_{xx}
	\end{equation}
	in the form
	\begin{equation}\label{s1}
	u(x,t) = u(\xi) = \gamma_0 + \gamma_1 c(\xi) + \gamma_2 c(\xi)^2
	\end{equation}
	where $c(\xi)$ is solution of the simplest equation
	\begin{equation}\label{s2}
	\frac{dc}{d\xi} = d_0 + d_1 c(\xi) + d_2 c(\xi)^2 + d_3 c(\xi)^3
	\end{equation}
	\item[7.)] \emph{The system of algebraic equations and its solution}\\
	The substitution of Eqs.(\ref{s1}) and (\ref{s2}) transforms Eq.(\ref{ex1z}) to polynomial
	of $c(\xi)$. After setting the coefficients of this polynomial to $0$ we obtain a system of 9
	nonlinear algebraic equations. We shall not write this system here as it is quite long. One nontrivial solution of the system of algebraic equations is
	\begin{eqnarray}\label{s3}
	d_0 &=& \left[ \frac{\beta p (-\alpha \beta p)^{1/2}}{64 \alpha^5 q^3 d_3}\right]^{1/2}; \ \ 
	d_1 = -\left(-\frac{\beta p}{16 \alpha^3 q^2} \right)^{1/2}; \gamma_0 = 0; \nonumber \\
	d_2 &=& \frac{3}{2} \left( \frac{d_3}{\alpha^2 q} \right)^{1/2} \left( - \alpha \beta p \right)^{1/4};
	\ \ \gamma_2 = - \frac{4 \alpha q d_3}{p}; \ \ \gamma_1= - (q d_3)^{1/2} \left( - \frac{\alpha \beta}{p^3} \right)^{1/4}\nonumber \\
	\end{eqnarray}
	The simplest equation becomes
	\begin{eqnarray}\label{s4}
	\frac{dc}{d\xi} &=& \left[ \frac{\beta p (-\alpha \beta p)^{1/2}}{64 \alpha^5 q^3 d_3}\right]^{1/2}
	-\left(-\frac{\beta p}{16 \alpha^3 q^2} \right)^{1/2} c(\xi) + \nonumber \\
	&& \frac{3}{2} \left( \frac{d_3}{\alpha^2 q} \right)^{1/2} \left( - \alpha \beta p \right)^{1/4} c(\xi)^2 + d_3 c(\xi)^3
	\end{eqnarray}
	Eq.(\ref{s4}) is a particular case of the differential equation of Abel of the first kind.
	This solution can be expressed by elementary functions for the particular case when 
	$d_0 = d_2 \frac{d_1 - \frac{2d_2^2}{9 d_3}}{3 d_3}$. As this relationship is fulfilled here 
	then the  solution of the Abel equation becomes
	\begin{eqnarray}\label{s6}
	c(\xi) = \frac{\exp \left[ \left( d_1 - \frac{d_2^2}{3 d_3} \right) \xi \right]}{\left \{ 
		C - \frac{d_3}{\left( d_1 - \frac{d_2^2}{d_3} \right)}  \exp \left [ 2 \left( d_1 - \frac{d_2}{3d_3} \right) \xi \right]   \right \}^{1/2}} - \frac{d_2}{3d_3},
	\end{eqnarray}
	where $C$ is a constant of integration. The solution of Eq.(\ref{ex1z}) then becomes
	\begin{eqnarray}\label{s7}
	&&u(x,t) = u(\xi) = \nonumber \\
	&&- (q d_3)^{1/2} \left( - \frac{\alpha \beta}{p^3} \right)^{1/4} \left \{\frac{\exp \left[ \left( d_1 - \frac{d_2^2}{3 d_3} \right) \xi \right]}{\left \{ 
		C - \frac{d_3}{\left( d_1 - \frac{d_2^2}{d_3} \right)}  \exp \left [ 2 \left( d_1 - \frac{d_2}{3d_3} \right) \xi \right]   \right \}^{1/2}} - 
	\frac{d_2}{3d_3} \right \} - \nonumber\\
	&&\frac{4 \alpha q d_3}{p} \left \{ \frac{\exp \left[ \left( d_1 - \frac{d_2^2}{3 d_3} \right) \xi \right]}{\left \{ 
		C - \frac{d_3}{\left( d_1 - \frac{d_2^2}{d_3} \right)}  \exp \left [ 2 \left( d_1 - \frac{d_2}{3d_3} \right) \xi \right]   \right \}^{1/2}} - \frac{d_2}{3d_3}\right \}^2
	\end{eqnarray}
	where $d_{0,1,2,3}$ are given by Eqs.(\ref{s3}).
\end{description}
\subsection{Solution for a case of fractional values of powers $a$ and $b$ in Eq.(\ref{ex1})}
Below we shall consider traveling-wave solutions 
\begin{math}
u(x,t) = u( \xi ) = u(\alpha x + \beta t), 
\end{math} constructed on the basis of the simplest equation 
\begin{equation}\label{se}
f_{\xi} = n \left[ f^{(n-1)/n} - f^{(n+1)/n} \right],
\end{equation}
where $n$ is \emph{an appropriate positive real number}.
The solution of this equation is $f( \xi ) = \tanh^{n}(\xi)$.
$n$ must be such real number that $\tanh^{n}(\xi)$ exists for $\xi \in (-\infty,
+ \infty)$ ($n=1/3$ is an appropriate value for $n$ and $n=1/2$ is not an
appropriate value for $n$).
\par 
Let us now prove a theorem before proceeding with application of the discussed version of the
modified method of simplest equation.
\begin{theorem}
	Let {$\cal{P}$} be a polynomial of the function $u(x,t)$ and its derivatives. $u(x, t)$ can be differentiated $k$ times, where $k$ is
	the highest order of derivative participating in {$\cal{P}$}. 
	Let us consider the nonlinear partial differential equation:
	\begin{equation}\label{nde}
	{\cal{P}}=0
	\end{equation}
	We search for solutions of this equation of the kind 
	$u( \xi) = \gamma + \delta f( \xi ), \xi = \alpha x + \beta t$. $\gamma $ and
	$\delta$  are parameters and $f( \xi )$ is a solution of the 
	simplest equation $f_{\xi} = n \left[ f^{(n-1)/n} - f^{(n+1)/n} \right]$ 
	where $n$ is an appropriate real positive number.
	The substitution of this solution in Eq.(\ref{nde}) leads to a relationship R of the kind
	\begin{equation}\label{r}
	R = \sum_{\sigma_i} \kappa_{\sigma_i} f^{\sigma_i}
	\end{equation}
	where $\sigma_i$ are some real numbers  and $\kappa_{\sigma_i}$ are algebraic
	relationships containing the parameters of the solved equation and the parameters of the 
	solution. Any nontrivial solution of the system of (nonlinear) algebraic equations
	$\kappa_{\sigma_i}=0$, $i=1,\dots$ leads to a solution of the solved nonlinear partial differential equation. 
\end{theorem}
\begin{proof}
	Let us denote the $k$-th derivative of $f(\xi)$ as $f^{(k)}_\xi$. First we shall
	prove that if $f(\xi)$ obeys Eq.(\ref{se}) then 
	\begin{equation}\label{rel}
	f_\xi^{(k)} = \sum \limits_{m=0}^k
	g_m(n) f^{(n-k+2m)/n}.
	\end{equation}
	where $g_m(n)$ is a polynomial of $n$.
	In order to proof this we mention that $f_\xi$ satisfies
	this relationship as it can be seen from Eq.(\ref{se}). The second, third, and
	fourth derivatives of $f(\xi)$
	\begin{eqnarray}\label{ders}
	f_{\xi \xi} &=& n \bigg[(n-1) f^{(n-2)/n} -2nf + (n+1)f^{(n+2)/2}\bigg]
	\nonumber \\
	f_{\xi \xi \xi} &=& n \bigg[ (n^2-3n+2)f^{(n-3)/n} + (-3n^2 + 3n -2)F^{(n-1)/n} +
	\nonumber \\
	&& (3n^2+3n+2)f^{(n+1)/n} - (n^2+3n+2)f^{(n+3)/3}\bigg] \nonumber \\
	f_{\xi \xi \xi \xi} &=& n \bigg[ -(n^3-6n^2+11n-6)f^{(n-4)/n} -
	4(n^3-3n^2+4n-2)f^{(n-2)/n} + \nonumber \\
	&& (6n^3+10n)f - 4(n^3+3n^2+4n+2)f^{(n+2)/n} + \nonumber \\
	&& (n^3+6n^2+11n+6)f^{(n+4)/n}\bigg]
	\end{eqnarray}
	are of the kind (\ref{rel}). Let us assume that the $k$-th derivative of
	$f(\xi)$ is of kind (\ref{rel}). We shall show that the $(k+1)$-st derivative 
	of $f(\xi)$ is of kind (\ref{rel}). From Eq.(\ref{rel})  we obtain ($q=k+1$)
	\begin{equation}\label{rel1}
	f_\xi^{(k+1)} = \sum \limits_{m=0}^{q-1}g_m(n)(n-q+2m+1) f^{(n-q+2m)/n} -
	\sum \limits_{m=0}^{q-1} g_m(n) (n-q+2m+1) f^{(n-q+2m+2)/n}
	\end{equation}
	The term $\sum \limits_{m=0}^{q-1}g_m(n)(n-q+2m+1) f^{(n-q+2m)/n}$ is of
	kind (\ref{rel}). Let us split the term 
	$\sum \limits_{m=0}^{q-1} g_m(n) (n-q+2m+1) f^{(n-q+2m+2)/n}$ in two parts:
	$\sum \limits_{m=0}^{q-2} g_m(n) (n-q+2m+1) f^{(n-q+2m+2)/n}$ and 
	$\sum \limits_{m=q-1}^{q-1} g_m(n) (n-q+2m+1) f^{(n-q+2m+2)/n}$. The first
	term above is of kind (\ref{rel}). Thus for $f^{(k+1)}_\xi$ we have up to now
	all terms of the sum (\ref{rel}) except the last one (the term corresponding to 
	$m=q$ that must contain
	$f^{(n+q)/n}$). But let consider the the second term above. This is exactly
	the missing term as it is equal to $g_{q-1}(n)
	(n+q-1)f^{(n+q)/n}$ which is term of the same kind as the $k$-th term of the sum
	in Eq.(\ref{rel}). Thus the derivative $f_\xi^{(k+1)}$ is also of the kind
	(\ref{rel}) and our proposition is proven by the method of mathematical
	induction.
	\par 
	We have shown that $f_\xi^{(k)}$ is of the kind of the terms in (\ref{r}). Then it follows that
	any of the terms of the solved equation (\ref{nde}) is of the same kind or of the kind $f^0$. 
	Thus {\cal{P}} is is reduced to
	relationship of kind (\ref{r}). Then any nontrivial solution of the system of nonlinear 
	algebraic equations $\kappa_{\sigma_i}=0$, $i=1,\dots$ leads to a solution of the solved 
	nonlinear partial differential equation ${\cal{P}}=0$.
\end{proof}
\par 
Let us now apply the SEsM based on 
the simplest equation (\ref{se}) to the equation (\ref{ex1}). The steps of the method are as follows.
\begin{figure}[!htb]
	\centering
	\includegraphics[scale=0.8]{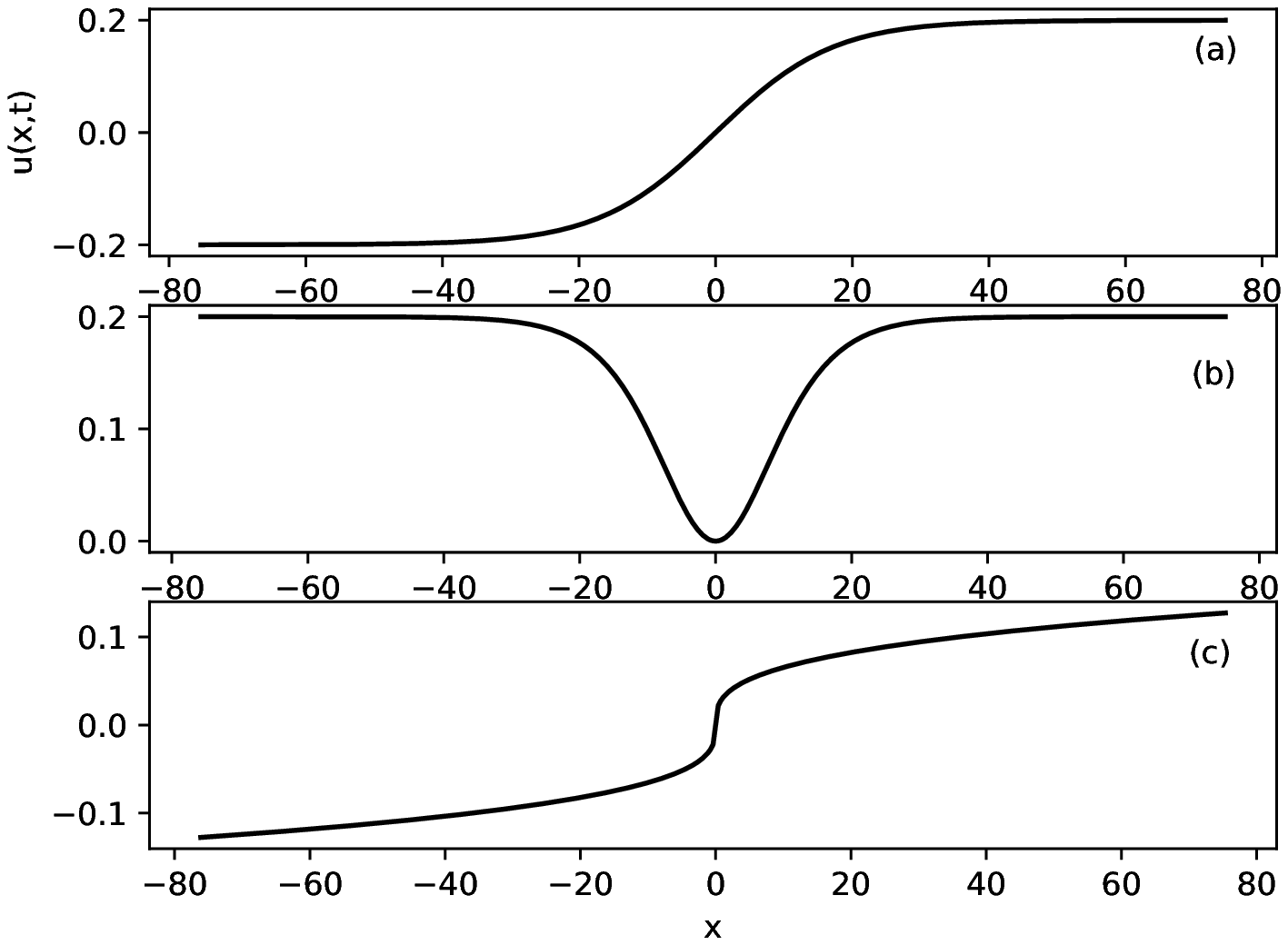}
	\caption{Several profiles of the solutions of Eq.(\ref{ex1}). The values of parameters
		are: $\delta=0.2$, $p=0.7$, $q=1.2$, $t=0.0$. Figure (a): The solution (\ref{kink1}). 
		Figure (b): The solution (\ref{kink2}). Figure (c): The solution (\ref{kink3}).}
\end{figure}

\begin{description}
	\item[1.)] \emph{The transformation}\\
	We shall  use a particular case of transformation (\ref{m1}) of $u(x,t)$, i.e. $u(x,t)=F(x,t)$
	\item[2.)] \emph{Relationship among $F(x,t)$  and the functions functions $f_k(x,t)$} \\
	We shall use a single function $f(x,t)$ and the following particular case of Eq.(\ref{m2})
	\begin{equation}\label{ez1}
	F(x,t) = \delta f(x,t),
	\end{equation}
	where $\delta$ is a  parameter.
	\item[3.)] \emph{Representation of $f(x,t)$ by a solution of simplest equation}\\
	We shall use the most simple expression containing just one function that is a traveling wave: 
	$f(x,t) = e(\xi)$;  $\xi = \alpha x + \beta t$
	\item[4.)] \emph{Expression of $e(\xi)$ by function that is solutions of a simplest equation}\\
	We shall use again the most simple relationship $e(\xi) = c(\xi)$.
	\item[5.)] \emph{The simplest equation}\\
	The simplest equation for $c(\xi)$ is
	\begin{equation}\label{sec}
	c_{\xi} = n \left[ c^{(n-1)/n} - c^{(n+1)/n} \right].
	\end{equation}
	where $n$ can be a positive real number.
	\item[6.)] \emph{Balance equations}\\
	The substitution of the equations from the steps 2.) - 5.) transforms Eq.(\ref{ex1}) to a
	polynomial of $c(\xi)$. Note that the conditions of the above theorem are satisfied.
	The application of the balance leads to
	the relationships $a=1+2/n$, $b= 1 +1/n$ (two balance equations). 
	\item[7.)] \emph{The system of algebraic equations and its solutions}\\
	The use of the balance equations leads  to  the following system of nonlinear algebraic equations
	\begin{eqnarray}\label{nae1}
	\beta  - (n+1) q \alpha^2 \delta^{1/n} = 0 \nonumber \\
	p  \delta^{1/n} - q \alpha (n+1)  = 0 \nonumber \\
	2 q \alpha^2 (n+1) \delta^{1/n} - \alpha p (n+2) \delta^{2/n} + 2n[q
	\alpha^2(n+1) \delta^{1/n} - \beta/2] = 0 \nonumber \\
	\end{eqnarray}
	The solution of Eqs(\ref{nae1}) is
	\begin{equation}\label{sl1}
	\alpha =  \frac{p}{(n+1)q} \delta^{1/n}; \ \ \ 
	\beta = \frac{p^2}{(n+1)q} \delta^{3/n} 
	\end{equation}
	and then the solution of the equation (\ref{ex1}) ($a=1+2/n$, $b=1+1/n$)
	\begin{equation}\label{ex1a}
	u_t = p\Bigg(1+\frac{2}{n}\Bigg)u^{2/n}u_x + q \frac{1}{n}\Bigg( 1+ 
	\frac{1}{n} \Bigg) u^{(1/n)-1}u_x^2 + q \Bigg(1+\frac{1}{n} \Bigg) u^{1/n}
	u_{xx}
	\end{equation}
	is 
	\begin{equation}\label{sl1}
	u(x,t) = \delta \tanh^n \left[ \frac{p}{(n+1)q} \delta^{1/n} \Bigg( x + 
	p \delta^{2/n} t \Bigg) \right]
	\end{equation}
	\par
	Let us consider several particular cases. For $n=1$: the equation
	\begin{equation}\label{ks1}
	u_t = 3pu^2u_x + 2qu_x^2 + 2quu_{xx},
	\end{equation}
	has the solution
	\begin{equation}\label{kink1}
	u(x,t) = \delta \tanh \left[ \frac{p}{2q} \delta \Bigg( x + 
	p \delta^{2} t \Bigg) \right].
	\end{equation}
	For $n=2$: the equation
	\begin{equation}\label{ks2}
	u_t = 2puu_x + \frac{3q}{4}u^{-1/2}u_x^2 + \frac{3q}{2}u^{1/2} u_{xx},
	\end{equation}
	has the solution
	\begin{equation}\label{kink2}
	u(x,t) = \delta \tanh^2 \left[ \frac{p}{3q} \delta^{1/2} \Bigg( x + 
	p \delta^4 t \Bigg) \right].
	\end{equation}
	For $n=1/3$: the equation
	\begin{equation}\label{ks3}
	u_t=7p u^6 u_x + 12 q u^2u_x^2 + 4 q u^3u_{xx},
	\end{equation}
	has the solution
	\begin{equation}\label{kink3}
	u(x,t) = \delta \tanh^{1/3} \left[ \frac{3p}{4q} \delta^{3} \Bigg( x + 
	p \delta^{6} t \Bigg) \right].
	\end{equation}
	\begin{figure}[!htb]
		\centering
		\includegraphics[scale=.8]{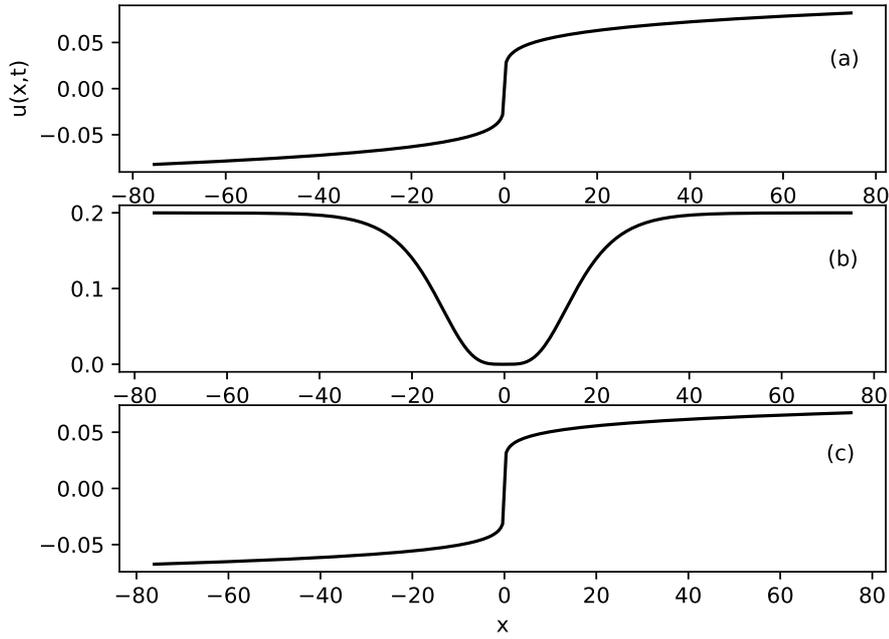}
		\caption{Several profiles of the solutions of Eq.(\ref{ex1}). The values of parameters
			are: $\delta=0.2$, $p=0.7$, $q=1.2$, $t=0.0$. Figure (a): The solution (\ref{kink4}). $m=5$.
			Figure (b): The solution (\ref{kink4}). $m=1/4$. Figure (c): The solution (\ref{kink4}). $m=7$.}
	\end{figure}
	For $n=1/m$ (m is an odd integer): the equation
	\begin{equation}\label{ks4}
	p(1+2m)u^{2m}u_x + qm(m+1)u^{m-1}u_x^2 + q(m+1)u^mu_{xx},
	\end{equation}
	has the solution
	\begin{equation}\label{kink4}
	u(x,t) = \delta \tanh^{1/m} \left[ \frac{mp}{(m+1)q} \delta^{m} \Bigg( x + 
	p \delta^{2m} t \Bigg) \right].
	\end{equation}
	\par 
	Several profiles of solutions obtained above are presented in Fig. 5 and Fig. 6.
\end{description}
\section{Concluding remarks}
We discuss in this article a methodology for obtaining exact analytical solutions of nonlinear partial differential equations. The methodology is called Simple equations Method (SEsM) and it is based on the possibility of use of more than one simple
equation. We add a possibility for a transformation connected to  the searched solution. In such a way the possibility for use 
of a Painleve expansion or other transformations  is presented in the methodology. This possibility 
in combination with the possibility of use of more than one simple equation adds
the capability for obtaining multisolitons by the SEsM. In addition we consider
the relationship (\ref{m2}) that is used to connect the solution of the solved 
nonlinear partial differential equation to solutions of more simple differential 
equations. The relationship (\ref{m2}) contains as particular case the relationship
used by Hirota \cite{hirota}. In addition the relationship (\ref{m2}) contains as
particular case the relationship used in the previous versions of the methodology based on 1 simple equation (and called  Modified 
method of simplest equation) for connection of the searched solution of the solved
nonlinear partial differential equation to the solution of the used simplest equation.
The discussed version of the methodology allows for the use of more than one 
balance equation too. It is demonstrated that the SEsM can indeed lead to multisoliton solutions. In addition 
the discussed version of the methodology preserves its capability for obtaining
exact particular solutions of non-integrable nonlinear partial differential equations.
Several examples of application of the discussed version of the methodology are 
presented and it is demonstrated that the balance procedure can lead to more than one 
balance equation. Special attention is given to the application of the methodology
on the basis of the simplest equation $f_\xi =n[f^{(n-1)/n} - f^{(n+1)/n}]$  where
$n$ can be a positive real number. Several solutions of partial differential equations containing
fractional powers are obtained on the basis of this simplest equation.


\end{document}